\documentclass[12pt,a4paper]{article}%
\usepackage{amssymb}
\usepackage{amsfonts}
\usepackage{amsmath}
\usepackage{graphicx}
\usepackage[singlespacing]{setspace}
\usepackage{geometry}%
\providecommand{\U}[1]{\protect\rule{.1in}{.1in}}
\newtheorem{theorem}{Theorem}

\newtheorem{proposition}[theorem]{Proposition}

\newenvironment{proof}[1][Proof]{\noindent\textbf{#1.} }{\ \rule{0.5em}{0.5em}}
\begin{document}

\title{\textbf{Ramsey Optimal Policy versus Multiple Equilibria with Fiscal and
Monetary Interactions\bigskip}\\Post Print: Economics Bulletin, 40(1), pp.140-147. }
\author{Jean-Bernard Chatelain\thanks{Paris School of Economics, Universit\'{e} Paris
I Pantheon Sorbonne, 48 Boulevard Jourdan, 75014\ Paris. Email:
jean-bernard.chatelain@univ-paris1.fr} and Kirsten Ralf\thanks{ESCE
International Business School, INSEEC\ U. Research Center, 10 rue Sextius
Michel, 75015 Paris, Email: Kirsten.Ralf@esce.fr.}}
\maketitle

\begin{abstract}
We consider a frictionless constant endowment economy based on Leeper (1991).
In this economy, it is shown that, under an ad-hoc monetary rule and an ad-hoc
fiscal rule, there are two equilibria. One has active monetary policy and
passive fiscal policy, while the other has passive monetary policy and active
fiscal policy. We consider an extended set-up in which the policy maker
minimizes a loss function under quasi-commitment, as in Schaumburg and
Tambalotti (2007). Under this formulation there exists a unique Ramsey
equilibrium, with an interest rate peg and a passive fiscal policy.

\textbf{JEL\ classification numbers}: E63, C61, C62, E43, E44, E47, E52, E58.


\end{abstract}

\section{Introduction}

Monetary and fiscal interactions are often presented using Leeper's (1991)
model of frictionless endowment economies. This model includes the
intertemporal budget constraint of the government and a Fisher relation with a
constant real interest rate equal to the representative household's discount
factor. In Leeper (1991), for monetary policy, the interest rate responds in
proportion to inflation. For fiscal policy, a lump-sum tax responds in
proportion to the stock of real public debt.

Assuming inflation, interest rate and lump-sum tax are forward-looking
variables for seeking equilibria with Blanchard and Kahn's (1980) determinacy
condition, Leeper (1991) obtains two equilibria for monetary and fiscal
interactions. These two equilibria are defined by ranges of values of the
parameters of ad hoc policy rules. In the first equilibrium, the interest rate
rule parameter destabilizes inflation ("active monetary policy" according to
Leeper (1991)) and the fiscal rule parameter stabilizes public debt ("passive
fiscal policy" according to Leeper (1991)). In the second equilibrium related
to the fiscal theory of the price level, the fiscal rule parameter
destabilizes public debt ("active policy" according to Leeper (1991)) and the
interest rate rule parameter stabilizes inflation ("passive monetary policy"
according to Leeper (1991)). A peg of the interest rate with a lack of
response of the interest rate to inflation is a particular case of a "passive
monetary policy" equilibrium.

This note demonstrates that Ramsey optimal policy is a third equilibrium with
an optimal interest rate peg and "passive" fiscal policy.

Section two presents the policy transmission mechanism, the two equilibria
with ad hoc policy rules and the Ramsey optimal policy equilibrium. Section
three concludes.

\section{Ramsey Optimal Policy in a Frictionless Endowment Economy}

\subsection{Policy Transmission Mechanism}

Bai and Leeper (2017) and Cochrane (2019, chapter 2) omit money in the
frictionless endowment economy model with respect to Leeper's (1991) seminal
paper, while still obtaining Leeper's (1991) two regimes for monetary and
fiscal interactions.

Consider an infinitely-lived representative consumer who receives a
\emph{constant} endowment of goods each period in the amount $y$ and derives
utility only from consumption $c_{t}$. The government purchases a
\emph{constant} quantity of goods $g>0$ from the consumer each period. We
impose equilibrium in the goods market, so that consumption: $c_{t}$ $=y-g$.

The consumer makes a consumption-saving decision that produces the Fisher
relation where the real rate is here equal to a \emph{constant} discount rate:%

\begin{equation}
E_{t}\frac{1}{\pi_{t+1}}=\frac{1}{\beta}\frac{1}{R_{t}}%
\end{equation}

where $R_{t}$ is both the gross one-period nominal interest rate on nominal
bonds bought at $t$ and pay off in $t+1$ and the monetary policy instrument,
and $\pi_{t+1}=P_{t+1}/P_{t}$ is the gross rate of inflation between $t$ and
$t+1$, with $P_{t}$ the aggregate price level. In the steady state,
$P_{t}=P_{t+1}\Rightarrow\pi^{\ast}=1$ where $\pi^{\ast}$ is the inflation
target. Then, $R^{\ast}=\pi^{\ast}/\beta=1/\beta$, is the nominal interest
rate consistent with the inflation target according to the\ Fisher relation.
The equilibrium real interest rate is constant at $r=(1/\beta)-1$ where
$0<\beta<1$ is the consumer's discount factor. The Fisher relation in
deviation of steady state values is:%

\begin{equation}
E_{t}\frac{1}{\pi_{t+1}}-\frac{1}{\pi^{\ast}}=\frac{1}{\beta}\left(  \frac
{1}{R_{t}}-\frac{1}{R^{\ast}}\right)  \Rightarrow E_{t}\frac{1}{\pi_{t+1}%
}-1=\frac{1}{\beta}\left(  \frac{1}{R_{t}}-\beta\right)  \text{.}%
\end{equation}

It is linearized around the steady state equilibrium:%

\[
E_{t}\pi_{t+1}-\pi^{\ast}=\beta\left(  R_{t}-R^{\ast}\right)  \Rightarrow
E_{t}\pi_{t+1}-1=\beta\left(  R_{t}-\frac{1}{\beta}\right)  .
\]

Fiscal policy levies lump-sum taxes of $\tau_{t}$ and sets purchases to be
constant, $g>0$ with primary surplus $s_{t}=\tau_{t}-g$. Government issues
one-period nominal bonds, $B_{t}$, that satisfy the flow constraint, where
$P_{t}$ is the aggregate price level and real debt is defined as $b_{t}%
=B_{t}/P_{t}$.%

\[
b_{t}=\frac{B_{t}}{P_{t}}=-\left(  \tau_{t}-g\right)  +R_{t-1}\frac{P_{t-1}%
}{P_{t}}\frac{B_{t-1}}{P_{t-1}}=-\left(  \tau_{t}-g\right)  +\frac{R_{t-1}%
}{\pi_{t}}b_{t-1}.
\]

Using the Fisher relation, we substitute the constant real interest rate in
the government intertemporal budget constraint:%

\[
b_{t}=-\left(  \tau_{t}-g\right)  +\frac{1}{\beta}b_{t-1}.
\]

The dynamics of real debt does not depend on inflation or on the nominal rate
$R_{t}$. The steady state level of real government debt has an exogenous
value: $b_{t+1}=b_{t}=b^{\ast}$. To be consistent with the government
intertemporal budget constraint, the steady state level of tax revenue
$\tau^{\ast}$ is equal to the steady state interest expense:%

\[
\tau^{\ast}-g=\left(  \frac{1}{\beta}-1\right)  b^{\ast}.
\]

The government intertemporal budget constraint written in deviation of steady
state is%

\[
b_{t}-b^{\ast}=-\left(  s_{t}-s^{\ast}\right)  +\frac{1}{\beta}(b_{t-1}%
-b^{\ast}).
\]

The linearized dynamics are (in deviation from steady state):%

\[
\left(
\begin{array}
[c]{c}%
E_{t}\pi_{t+1}-\pi^{\ast}\\
E_{t}b_{t+1}-b^{\ast}%
\end{array}
\right)  =\underset{=\mathbf{A}}{\underbrace{\left(
\begin{array}
[c]{cc}%
A_{\pi}=0 & A_{\pi b}=0\\
A_{b\pi}=0 & A_{b}=\frac{1}{\beta}%
\end{array}
\right)  }}\left(
\begin{array}
[c]{c}%
\pi_{t}-\pi^{\ast}\\
b_{t}-b^{\ast}%
\end{array}
\right)  +\underset{=\mathbf{B}}{\underbrace{\left(
\begin{array}
[c]{cc}%
B_{\pi R}=\beta & B_{\pi s}=0\\
B_{bR}=0 & B_{\pi s}=-1
\end{array}
\right)  }}\left(
\begin{array}
[c]{c}%
R_{t}-R^{\ast}\\
s_{t+1}-s^{\ast}%
\end{array}
\right)  .
\]

The log-linearized dynamics of inflation is:%

\[
\frac{E_{t}\pi_{t+1}-\pi^{\ast}}{\pi^{\ast}}=\beta\left(  \frac{R^{\ast}}%
{\pi^{\ast}}\right)  \frac{R_{t}-R^{\ast}}{R^{\ast}}=\frac{R_{t}-R^{\ast}%
}{R^{\ast}}.
\]

The log-linearized dynamics of real debt is:%

\[
\frac{E_{t}b_{t+1}-b^{\ast}}{b^{\ast}}=\frac{1}{\beta}\frac{b_{t}-b^{\ast}%
}{b^{\ast}}-\left(  \frac{s^{\ast}}{b^{\ast}}\right)  \frac{s_{t+1}-s^{\ast}%
}{s^{\ast}}\text{ with }\frac{\tau^{\ast}-g}{b^{\ast}}=\frac{1}{\beta}-1.
\]

Log-linearized dynamics are:%

\[
\left(
\begin{array}
[c]{c}%
\frac{E_{t}\pi_{t+1}-\pi^{\ast}}{\pi^{\ast}}\\
\frac{E_{t}b_{t+1}-b^{\ast}}{b^{\ast}}%
\end{array}
\right)  =\underset{=\mathbf{A}}{\underbrace{\left(
\begin{array}
[c]{cc}%
A_{\pi}=0 & A_{\pi b}=0\\
A_{b\pi}=0 & A_{b}=\frac{1}{\beta}%
\end{array}
\right)  }}\left(
\begin{array}
[c]{c}%
\frac{\pi_{t}-\pi^{\ast}}{\pi^{\ast}}\\
\frac{b_{t}-b^{\ast}}{b^{\ast}}%
\end{array}
\right)  +\underset{=\mathbf{B}_{\log}}{\underbrace{\left(
\begin{array}
[c]{cc}%
1 & 0\\
0 & -\left(  \frac{1}{\beta}-1\right)
\end{array}
\right)  }}\left(
\begin{array}
[c]{c}%
\frac{R_{t}-R^{\ast}}{R^{\ast}}\\
\frac{s_{t+1}-s^{\ast}}{s^{\ast}}%
\end{array}
\right)
\]

\subsection{Ad Hoc Policy Rules}

The fiscal authority adjusts lump-sum tax in response to the level of real
government debt. Monetary policy follows an interest rate rule that responds
to inflation:%

\[
s_{t}-s^{\ast}=G_{b}\left(  b_{t-1}-b^{\ast}\right)  +\varepsilon_{t}%
^{s}\text{ and }R_{t}-R^{\ast}=F_{\pi}\left(  \pi_{t}-\pi^{\ast}\right)
+\varepsilon_{t}^{R}%
\]

Shocks $\varepsilon_{t}^{R}$, $\varepsilon_{t}^{s}$ are assumed to be
independently and identically distributed, with mean zero and a non-zero
variance-covariance matrix. The linearized dynamics are:%

\[
\left(
\begin{array}
[c]{c}%
E_{t}\frac{1}{\pi_{t+1}}-\frac{1}{\pi^{\ast}}\\
E_{t-1}b_{t+1}-b^{\ast}%
\end{array}
\right)  =\left(
\begin{array}
[c]{cc}%
\frac{1}{\beta}F_{\pi} & 0\\
0 & \frac{1}{\beta}-G_{b}%
\end{array}
\right)  \left(
\begin{array}
[c]{c}%
\pi_{t}-\pi^{\ast}\\
b_{t}-b^{\ast}%
\end{array}
\right)  +\left(
\begin{array}
[c]{cc}%
\beta & 0\\
0 & -1
\end{array}
\right)  \left(
\begin{array}
[c]{c}%
\varepsilon_{t}^{R}\\
\varepsilon_{t}^{s}%
\end{array}
\right)  \text{ with }b_{0}\text{ given.}%
\]

Public debt is the only predetermined variable. Blanchard and Kahn's (1980)
determinacy condition implies that one eigenvalue should be inside the unit
circle and one should be outside the unit circle. Either the inflation
eigenvalue is outside the unit circle $\left\vert \beta F_{\pi}\right\vert >1$
and the public debt eigenvalue $\left\vert \frac{1}{\beta}-G_{b}\right\vert
<1$ is inside the unit circle (first equilibrium) or it is the reverse (second
equilibrium):\textbf{ }$\left\vert \beta F_{\pi}\right\vert <1$ and
$\left\vert \frac{1}{\beta}-G_{b}\right\vert >1$ (Leeper (1991)).

\subsection{Ramsey Optimal Policy under Quasi-Commitment}

In a monetary policy regime indexed by $j$, a policy maker may re-optimize on
each future period with exogenous probability $1-q$ strictly below one.
Following Schaumburg and Tambalotti (2007), we assume that the mandate to
minimize the loss function is delegated to a sequence of policy makers with a
commitment of random duration. The degree of credibility is modelled as if it
is a change of policy-maker with a given probability of reneging commitment
and re-optimizing optimal plans. The length of their tenure or "regime"
depends on a sequence of exogenous independently and identically distributed
Bernoulli signals $\left\{  \eta_{t}\right\}  _{t\geq0}$ with $E_{t}\left[
\eta_{t}\right]  _{t\geq0}=1-q$, with $0<q\leq1$. If $\eta_{t}=1,$ a new
policy maker takes office at the beginning of time $t$. Otherwise, the
incumbent stays on. A higher probability $q$ can be interpreted as a higher
credibility. A policy maker with little credibility does not give a large
weight on future welfare losses. The policy maker $j$ solves the following
problem for regime $j$, omitting subscript $j$, before policy maker $k$ starts:%

\begin{align*}
V_{0}^{j}  &  =-E_{0}%
%TCIMACRO{\dsum \limits_{t=0}^{t=+\infty}}%
%BeginExpansion
{\displaystyle\sum\limits_{t=0}^{t=+\infty}}
%EndExpansion
\left(  \beta q\right)  ^{t}\left[
\begin{array}
[c]{c}%
\frac{1}{2}\left(  Q_{\pi}\left(  \pi_{t}-\pi^{\ast}\right)  ^{2}+Q_{b}\left(
b_{t}-b^{\ast}\right)  ^{2}+\mu_{R}\left(  R_{t}-R^{\ast}\right)  ^{2}+\mu
_{s}\left(  s_{t}-s^{\ast}\right)  \right) \\
+\beta\left(  1-q\right)  V_{t}^{k}%
\end{array}
\right] \\
R_{t}-R^{\ast}  &  =\beta q\left(  E_{t}\pi_{t+1}-\pi^{\ast}\right)
+\beta\left(  1-q\right)  \left(  E_{t}\pi_{t+1}^{k}-\pi^{\ast}\right) \\
b_{t}-b^{\ast}  &  =\beta q(E_{t}b_{t+1}-b^{\ast})+\beta\left(  1-q\right)
\left(  E_{t}b_{t+1}^{k}-b^{\ast}\right)  +\beta q(s_{t}-s^{\ast}%
)+\beta\left(  1-q\right)  (s_{t}^{k}-s^{\ast})\text{.}%
\end{align*}

Inflation and public debt expectations are an average between two terms. The
first term, with weight $q$ is the inflation and public debt, respectively,
that would prevail under the current regime upon which there is commitment.
The second term with weight $1-q$ is the inflation and public debt,
respectively, that would be implemented under the alternative regime by policy
maker $k$.

This optimal program is a discounted linear quadratic regulator (LQR) with a
"credibility adjusted" discount factor $\beta q$. The log-linear version of
the dynamic system can also be used instead of the linear version. Preferences
of the policy maker are firstly given by positive weights for the two policy
targets $Q_{\pi}\geq0$, $Q_{\pi}\geq0$ ($\mathbf{Q}=diag(Q_{\pi},Q_{b})$).
Secondly, there are at least two non-zero policy maker's preferences for
\emph{interest rate smoothing} and \emph{primary surplus and tax smoothing},
with \emph{strictly }positive weights for these two policy instruments in the
loss function: $\mu_{R}>0$, $\mu_{s}>0$ ($\mathbf{R}=diag(\mu_{R},\mu_{s})$).
This ensures the strict concavity of the LQR program which implies the
determinacy (uniqueness) of the solution of this Stackelberg dynamic game, if
the system of the transmission mechanism is controllable.

All the results are valid for nearly negligible cost of changing the policy
instruments: a minimal interest rate smoothing parameter ($\mu_{R}\geq10^{-7}%
$), a minimal tax smoothing parameter ($\mu_{s}\geq10^{-7}$) and a minimal
credibility (a non zero probability of not reneging commitment next period:
$10^{-7}\leq q\leq1$).

If $\mu_{R}=0$ or $\mu_{s}=0$ or if the costs of changing the policy
instruments are not strictly convex or if $q=0$, the results are no longer
valid. If $q=0$, the policy maker knows that he is replaced next period. He
does static optimization ($\left(  \beta q\right)  ^{0}=1$) of his current
period quadratic loss function subject to a static transmission mechanism
where he considers the expectations terms related to the next period policy
maker as exogenous intercepts (Chatelain and Ralf (2019b)).

\begin{proposition}
For the transmission mechanism of the Fisher relation with constant real rate
and the government intertemporal budget constraint, Ramsey optimal policy has
a unique equilibrium. The interest rate is pegged at its long run value
$R_{t}=R^{\ast}$. A feedback Taylor rule is not optimal and the Taylor
principle should not be satisfied: $F^{\ast}=0<\beta$. This is a "passive
monetary policy". The optimal auto-correlation of monetary policy shocks is
zero: $\rho_{\varepsilon^{R}}^{\ast}=0$. The optimal variance of monetary
policy shocks is zero $\sigma_{\varepsilon^{R}}^{2}=0$. Inflation jumps to its
steady state value instantaneously following monetary policy shocks $\pi
_{0}^{\ast}=0$, as in a degenerate rational expectations model without
predetermined variables. Hence, the price level is constant. Ramsey optimal
fiscal rule has a negative feedback parameter ("passive fiscal policy") with
ensures the local stability of public debt dynamics. There are two stable
eigenvalues giving the optimal persistence of inflation ($\lambda_{\pi}^{\ast
}=0$) and the optimal persistence of public debt ($0<\lambda_{b}^{\ast}<1$).
\end{proposition}

\begin{proof}
Ramsey optimal policy amounts to find the solution of a linear quadratic
regulator (Chatelain and Ralf (2019a)). The autocorrelation parameters of
policy maker's shocks are optimally set to zero, else they increase the
volatility of inflation and of public debt in the policy maker's loss
function. The optimal expected value of the loss function is:%
\begin{align}
L^{\ast}  &  =-\frac{1}{2}\left(
\begin{array}
[c]{cc}%
\widehat{\pi}_{0}^{\ast} & \widehat{b}_{0}%
\end{array}
\right)  \mathbf{P}\left(
\begin{array}
[c]{cc}%
\widehat{\pi}_{0}^{\ast} & \widehat{b}_{0}%
\end{array}
\right)  ^{T}\text{ with }\mathbf{P}=\left(
\begin{array}
[c]{cc}%
P_{\pi} & P_{\pi b}\\
P_{\pi b} & P_{b}%
\end{array}
\right)  \text{,}\\
\widehat{\pi}_{0}^{\ast}  &  =\pi_{0}^{\ast}-\pi^{\ast}\text{ and }%
\widehat{b}_{0}=b_{0}-b^{\ast}\text{.}%
\end{align}

where $\mathbf{P}$ is a positive symmetric square matrix of dimension two
which is the solution of the discrete algebraic Riccati equation (DARE) of the
LQR. The optimal initial anchor of inflation $\pi_{0}^{\ast}$ on public debt
$b_{0}$ is given by:%
\begin{equation}
\left(  \frac{\partial L^{\ast}}{\partial\pi_{t}}\right)  _{t=0}=P_{\pi
}\widehat{\pi}_{0}^{\ast}+P_{\pi b}\widehat{b}_{0}=0\Rightarrow\widehat{\pi
}_{0}^{\ast}=\frac{P_{\pi b}}{P_{\pi}}\widehat{b}_{0}\text{ if }P_{\pi}\neq0.
\end{equation}

This initial transversality (or natural boundary) condition eliminates the
indeterminacy of initial inflation $\pi_{0}$ put forward for the \textit{ad
hoc} policy rules solution. Because the system is decoupled ($A_{\pi
b}=A_{b\pi}=0$) and because the weight on the product $\pi_{t}b_{t}$ is zero
in the loss function $Q_{\pi b}=0$, this implies a zero weight $P_{\pi b}=0$
in the optimal value of the loss function. Therefore, $P_{\pi}\geq0$ and
$P_{b}\geq0$ are solutions of scalar discrete algebraic Riccati equations
(DARE). Because the system of the policy transmission mechanism is decoupled,
the optimal program is identical if we consolidate the central bank and the
treasury as a single policy maker or if we consider that they are distinct
policy makers. For inflation, this DARE\ equation is:
\[
P_{\pi}=Q_{\pi}+\beta A_{\pi}^{^{\prime}}P_{\pi}A_{\pi}-\beta^{^{\prime}%
}A_{\pi}^{^{\prime}}P_{\pi}B_{\pi R}\left(  \mu_{R}\mathbf{+}\beta B_{\pi
R}^{\prime}P_{\pi}B_{\pi R}\right)  ^{-1}\beta B_{\pi R}^{^{\prime}}P_{\pi
}A_{\pi}=Q_{\pi}.
\]

Because $A_{\pi b}=A_{b\pi}=A_{\pi}=0$, this implies that the optimal loss
function parameter $P_{\pi}$ is equal to the inflation weight in the loss
function: $P_{\pi}=Q_{\pi}$. Therefore, if the policy maker has a non-zero
weight on inflation volatility in his loss function $Q_{\pi}>0$, optimal
initial inflation is zero because $P_{\pi b}=0$: $\pi_{0}^{\ast}=0$. Because
inflation dynamics are decoupled from the dynamics of predetermined public
debt, inflation behaves exactly as in a degenerate rational equilibrium model
where there is no predetermined variable. At any date $t$, if ever there is a
monetary policy shock $\varepsilon_{t}^{R}$, inflation instantaneously jumps
back to equilibrium $\pi_{t}^{\ast}=0$. There are no transitory dynamics. The
volatility of inflation is zero.

Because the open-loop system is already at the zero lower bound of inflation
persistence ($A_{\pi}=0$) and because there is a \emph{non-zero} cost of
interest rate volatility ($\mu_{R}>0$) in the loss function, it is optimal to
set an interest rate peg $R_{t}=R^{\ast}$. Therefore, the quadratic term of
interest rate volatility $\mu_{R}\left(  R_{t}-R^{\ast}\right)  ^{2}$ is
minimized at its zero lower bound zero in the loss function. This implies a
Taylor rule parameter equal to zero: $F^{\ast}=0<\beta$, a zero
auto-correlation of monetary policy shocks $\rho_{\varepsilon^{R}}=0$, and a
zero variance of monetary policy shocks $\sigma_{\varepsilon^{R}}^{2}=0$
initiated by the policy maker.

If there is a zero cost of inflation in the loss function ($Q_{\pi}=0$), the
policy maker does not care about inflation. There is an indeterminacy for the
optimal choice of the initial value of inflation $\pi_{0}^{\ast}$. However,
the implied volatility originated by this indeterminacy does not matter in the
policy maker's loss function.

For public debt, the solution is a scalar case of Ramsey optimal policy under
quasi-commitment (we use the linear dynamics instead of the log-linear
dynamics).
\[
E_{t}b_{t+1}-b^{\ast}=\frac{1}{\sqrt{\beta q}}\left(  b_{t}-b^{\ast}\right)
-\sqrt{\beta q}(s_{t}-s^{\ast})\text{ with }A_{b}=\frac{1}{\beta q}\text{ or
}\beta qA=1\text{ and }B_{bs}=-1.
\]

Optimal public debt persistence (or auto-correlation or closed-loop
eigenvalue) is the stable root of the characteristic polynomial of the
Hamiltonian system:%
\begin{align*}
0  &  =\lambda^{2}-S\lambda+\frac{1}{\beta q}\text{ with }\lambda^{\ast
}\lambda_{2}=\frac{1}{\beta q}\text{ with}\\
S  &  =A+\frac{1}{A\beta q}+\frac{B^{2}}{A}\frac{Q_{b}}{\mu_{s}}=1+\frac
{1}{\beta q}+\beta q\frac{Q_{b}}{\mu_{s}}\\
0  &  <\lambda_{b}^{\ast}\left(  \mu_{s},\beta q\right)  =\frac{1}{2}\left(
S-\sqrt{S^{2}-\frac{4}{\beta q}}\right)  <1.
\end{align*}

The optimal public debt persistence or autocorrelation $\lambda_{b}^{\ast}$
increases with the cost $\mu_{s}$ of changing the policy instrument (lump sum
taxes). Its boundaries are given by limits when the cost of changing policy
instrument $\mu_{s}$ tends to zero and when tends $\mu_{s}$ to infinity:%
\[
\underset{\mu_{s}\rightarrow0}{\lim}\lambda_{b}^{\ast}\left(  \mu_{s},\beta
q\right)  =\underset{\mu_{b}\rightarrow0}{\lim}\frac{1}{2}\left(  \frac{\beta
q}{\mu_{s}}-\frac{\beta q}{\mu_{s}}\right)  =0\text{ and }\underset{\mu
_{s}\rightarrow+\infty}{\lim}\lambda_{b}^{\ast}\left(  \mu_{s},\beta q\right)
=\frac{1}{\beta qA}=1.
\]

The fiscal rule parameter $G_{b}$ remains in the range of values so that the
persistence of public debt is strictly positive and strictly below one:%
\[
\frac{1}{\beta q}-1<G_{b}^{\ast}=\frac{\lambda_{b}^{\ast}-A}{B}=\frac{1}{\beta
q}-\lambda_{b}^{\ast}<\frac{1}{\beta q}\text{.}%
\]

The optimal loss function parameter $P_{b}$ is:%
\[
P_{b}=\frac{\frac{1}{A}\frac{Q_{b}}{q\beta}}{1-\lambda_{b}^{\ast}}=\frac
{Q_{b}}{1-\lambda_{b}^{\ast}}\text{ .}%
\]

It is also the solution of a scalar algebraic Riccati quadratic equation. For
a given initial value of predetermined public debt $b_{0}$ and for preferences
$Q_{b}\geq0$ which can be equal to zero, the optimal expected value of the
policy maker loss function is:%
\begin{align*}
\text{If }Q_{\pi}  &  >0:\text{ }L^{\ast}=-\frac{1}{2}\left(
\begin{array}
[c]{cc}%
0 & \widehat{b}_{0}%
\end{array}
\right)  \left(
\begin{array}
[c]{cc}%
Q_{\pi} & 0\\
0 & \frac{Q_{b}}{1-\lambda_{b}^{\ast}}%
\end{array}
\right)  \left(
\begin{array}
[c]{cc}%
0 & \widehat{b}_{0}%
\end{array}
\right)  ^{T}=\text{ }-\frac{1}{2}\frac{Q_{b}}{1-\lambda_{b}^{\ast}}\left(
b_{0}-b^{\ast}\right)  ^{2}.\\
\text{ If }Q_{\pi}  &  =0:\text{ }L^{\ast}=-\frac{1}{2}\left(
\begin{array}
[c]{cc}%
\widehat{\pi}_{0}^{\ast} & \widehat{b}_{0}%
\end{array}
\right)  \left(
\begin{array}
[c]{cc}%
0 & 0\\
0 & \frac{Q_{b}}{1-\lambda_{b}^{\ast}}%
\end{array}
\right)  \left(
\begin{array}
[c]{cc}%
\widehat{\pi}_{0}^{\ast} & \widehat{b}_{0}%
\end{array}
\right)  ^{T}=\text{ }-\frac{1}{2}\frac{Q_{b}}{1-\lambda_{b}^{\ast}}\left(
b_{0}-b^{\ast}\right)  ^{2}.
\end{align*}

A\ rational policy maker with quadratic preferences including a convex cost of
changing policy instruments (interest rate smoothing $\mu_{R}>0$ and tax
smoothing $\mu_{s}>0$) taking into account private sector expectations with a
minimal credibility ($q>0$) eliminates the indeterminacy of initial inflation
$\pi_{0}$ if the policy maker's preferences includes a non-zero weight
($Q_{\pi}>0$) on inflation volatility in his loss function. Else, a rational
policy maker neglects the indeterminacy of initial inflation because inflation
volatility has zero weight ($Q_{\pi}=0$) in his loss function. A SCILAB code
is available from the authors to solve Ramsey optimal policy with numerical values.
\end{proof}

\section{Conclusion}

For the model of frictionless endowment economies, Ramsey optimal policy is an
interest rate peg and a "passive" fiscal policy. It is a third equilibrium
with respect to Leeper's (1991) two equilibria with simple feedback rules.


\begin{thebibliography}{9}                                                                                                %


\bibitem {Blanchard Kahn}Blanchard O.J. and Kahn C.\ (1980). The solution of
linear difference models under rational expectations. \textit{Econometrica,
}48, pp. 1305-1311.

\bibitem {Bai Leeper}Bai, Y., \& Leeper, E. M. (2017). Fiscal stabilization
vs. passivity. Economics Letters, 154, 105-108.

\bibitem {Chatelain Ralf}Chatelain, J. B., \& Ralf, K. (2019a). A Simple
Algorithm for Solving Ramsey Optimal Policy with Exogenous Forcing Variables.
Economics Bulletin. 39(4). pp. 2429-2440.

\bibitem {Chatelain Ralf 2}Chatelain, J. B., \& Ralf, K. (2019b). Imperfect
Credibility versus No Credibility of Optimal Monetary Policy. PSE working paper.

\bibitem {Cochrane 1}Cochrane J.H. (2019). \textit{The Fiscal Theory of the
Price Level}. Forthcoming book, version february 5. J.H. Cochrane's website.

\bibitem {Leeper}Leeper, E. M. (1991). Equilibria under `active'and
`passive'monetary and fiscal policies. Journal of monetary Economics, 27(1), 129-147.

\bibitem {Schaumburg Tambalotti}\textit{ }Schaumburg, E., \& Tambalotti, A.
(2007). An investigation of the gains from commitment in monetary policy.
Journal of Monetary Economics, 54(2), 302-324.
\end{thebibliography}
\end{document}